\documentclass[11pt]{article}
\pdfoutput=1 

\usepackage[utf8]{inputenc}
\usepackage[margin=2.5cm]{geometry}
\usepackage{amsfonts}
\usepackage{amsmath}
\usepackage{amssymb}
\usepackage{amsthm}
\usepackage{mathtools}
\usepackage{stmaryrd}
\usepackage{bold-extra} 
\usepackage{float}
\usepackage{microtype}
\usepackage{xspace}
\usepackage{bm}
\usepackage{doi}
\usepackage{hyperref}
\usepackage[dvipsnames]{xcolor}
\usepackage[capitalise, nameinlink, noabbrev]{cleveref}
\usepackage{tikz}
\usepackage{graphicx}
\usepackage[font=small,labelfont=bf]{caption} 
\usepackage{enumerate}
\usepackage{rotating}[figuresleft] 
\usepackage{stmaryrd}

\hypersetup{
	colorlinks=true,
	urlcolor=MidnightBlue,
	linkcolor=MidnightBlue,
	citecolor=MidnightBlue,
	unicode
}
\crefname{equation}{}{}
\crefname{claim}{Claim}{Claims}
\allowdisplaybreaks

\theoremstyle{plain}
\newtheorem{theorem}{Theorem}[section]
\newtheorem{lemma}[theorem]{Lemma}

\newtheorem{claim}[theorem]{Claim}
\newtheorem{informaltheorem}{Informal Theorem}

\theoremstyle{remark}
\newtheorem{remark}[theorem]{Remark}

\theoremstyle{definition}

\newcommand{\newclass}[2]{\newcommand{#1}{{\text{\upshape\sffamily #2}}\xspace}}
\newclass{\FP}{FP}
\newclass{\NP}{NP}
\newclass{\TFNP}{TFNP}
\newclass{\PLS}{PLS}
\newclass{\PPA}{PPA}
\newclass{\PPAD}{PPAD}
\newclass{\EOPL}{EOPL}
\newclass{\UEOPL}{UEOPL}

\newcommand{\newprob}[2]{\newcommand{#1}{{\text{\upshape\scshape #2}}\xspace}}
\newprob{\ueopl}{UEoPL}
\newprob{\eol}{EoL}
\newprob{\sod}{SoD}
\newprob{\Bracket}{Bracket}
\newprob{\BracketStar}{Bracket*}
\newprob{\CakeCutting}{CakeCutting}
\newprob{\JordanCurve}{JordanCurve}

\renewcommand{\epsilon}{\varepsilon}
\newcommand{\eps}{\epsilon}

\newcommand{\ZO}{\{0,1\}}

\newcommand{\cT}{\mathcal{T}}

\newcommand{\itemcomment}[1]{\hfill {\small #1}}

\newcommand{\checkR}{C^\text{R}}
\newcommand{\checkB}{C^\text{B}}
\DeclareMathOperator{\nextT}{next}

\newcommand{\poly}{\textup{\textsf{poly}}}

\title{On Cutting Cakes and Crossing Curves}

\author{%
  \makebox[.23\linewidth]{Alexandros Hollender}\\
  \textsl{University of Oxford}%
  \and
  \makebox[.23\linewidth]{Gilbert Maystre\footnote{Work done while at EPFL}}\\
  \textsl{Oracle Corporation}%
  \and \makebox[.23\linewidth]{Kilian Risse\footnote{Supported by the
      Swiss National Science Foundation
      Postdoc.Mobility fellowship \mbox{P500-2\_235298}. Part of this work done while at EPFL.}}\\
  \textsl{Lund University}%
}

\date{}

\begin{document}
\maketitle

\begin{abstract}
  We consider the classic envy-free cake-cutting problem where the
  goal is to cut and allocate a divisible resource among a set of
  agents in a way that avoids any envy between them. When the agents'
  valuation functions are continuous and nonnegative, an envy-free solution is
  guaranteed to exist where each agent is allocated a contiguous piece
  of the resource. Such a solution can be efficiently computed using the standard cut-and-choose algorithm for two agents, but the problem is known to be hard when there are at least four agents. The setting with three agents has remained open.

  We show that the problem remains intractable for three agents. We
  obtain this result by uncovering a novel connection between
  cake-cutting and a computational problem corresponding to the Jordan
  curve theorem, introduced by Adler, Daskalakis, and
  Demaine~(2016). As our main technical contribution, we provide the
  first lower bounds for the Jordan curve problem in the form of a query lower bound as well as hardness for the class $\UEOPL$, a subclass of $\PPAD$ containing notoriously challenging problems such as Simple Stochastic Games and the P\nobreakdash-matrix Linear Complementarity Problem.
\end{abstract}

\pagenumbering{roman}
\thispagestyle{empty}

\newpage
\pagenumbering{arabic}
\setcounter{page}{1}

\section{Introduction}

The cake-cutting problem, first studied by Steinhaus~\cite{Steinhaus1948-fair-division}, is perhaps the most fundamental problem in fair division~\cite{brams1996fair,robertson1998cake,procaccia2013cake}. We are given a divisible resource, henceforth called the \emph{cake} and modeled as the $[0,1]$ interval, and are asked to divide it in some fair manner among $n$ agents. Each of those $n$ agents has its own opinion about which parts of the cake are more or less valuable, and this is modeled by an agent-specific valuation function $v_i$ that assigns a real non-negative value to each interval of the cake. The most natural fairness notion is \emph{envy-freeness}, which requires that no agent envies the piece allocated to some other agent.

Perhaps surprisingly, as long as the valuation functions are continuous, an envy-free allocation is always guaranteed to exist~\cite{Stromquist80-existence,Woodall80-existence,Su99-rental-harmony}. Furthermore, this division is \emph{connected} (or \emph{contiguous}), meaning that each agent is allocated a single interval of the cake (as opposed to a union of multiple disjoint intervals). Unfortunately, the proof of this existence result relies on Sperner's lemma and does not yield an efficient algorithm for the problem. The state of the art for the computation of connected envy-free divisions can be summarized as follows:\footnote{For additional related work, in particular for different models of computation or different assumptions about the valuation functions, see, e.g.,~\cite{BranzeiN22-cake-query,BranzeiN19-cake-communication}.}
\begin{itemize}
\item For two agents, the problem can be solved efficiently by the \emph{cut-and-choose} algorithm: one agent cuts the cake in half (according to its own valuation function), and the other agent chooses its favorite piece, leaving the remaining piece to the agent who performed the cut. A cut position where the two resulting pieces have the same value for the agent can be (approximately) computed by binary search.
\item For four or more agents, the problem is intractable~\cite{HollenderR23-cake-four}. Namely, it requires a number of queries that is exponential in the bit complexity of the precision parameter, and is also \PPAD-complete in the standard Turing machine model.
\item If one assumes that the valuations are monotone,\footnote{A valuation $v$ is \emph{monotone}, if $v(A) \leq v(B)$ whenever $A \subseteq B$.} then the problem can be solved efficiently for up to four agents~\cite{HollenderR23-cake-four}. The complexity of monotone cake-cutting remains open for five or more agents.
\end{itemize}

Hence the main question left open for non-monotone valuations is to
determine the query-complexity of envy-free divisions with three
agents. In this work, we seek to close this gap in our understanding;
our main result is as follows.

\begin{informaltheorem}\label{inf:CC}
  Envy-free cake-cutting is intractable, even for three agents with
  identical (non-monotone) valuation functions. Namely, finding an
  $\eps$-envy-free division requires $\poly(1/\eps)$ queries
  and the problem is \UEOPL-hard in the standard Turing machine model.
\end{informaltheorem}

Attempts to extend techniques from prior work to prove hardness for three agents fail, because of the additional structure of the problem in that case. Indeed, when moving from four agents to three agents, the dimension of the space of solutions changes from three to two,\footnote{When there are three agents, we are looking to cut the cake in two positions, hence the space of potential solutions is two-dimensional.} and the two-dimensional space for this particular problem does not seem to allow a reduction from the canonical \PPAD-complete problem, which was used in prior work.

In order to bypass this obstacle, we consider a purely two-dimensional problem and establish a connection with our cake-cutting problem. The problem in question is a computational version of the Jordan curve theorem. Roughly, we are given two continuous curves in the square $[0,1]^2$, a red curve and a blue curve, such that (i) the red curve starts at $(0,0)$ and ends at $(1,1)$, and (ii) the blue curve starts at $(0,1)$ and ends at $(1,0)$. Clearly, these two curves must cross at some point, and the computational challenge is to find such a crossing point. This problem was first defined by Adler, Daskalakis, and Demaine~\cite{ADD16}, who proved that it lies in \PPAD and asked the question of determining its complexity.\footnote{The problem is called ``Crossing Curves'' in their paper.} As our main conceptual contribution, we establish the following connection.

\begin{informaltheorem}\label{inf:JC2CC}
The Jordan curve problem reduces to the envy-free cake-cutting problem with three agents with identical (non-monotone) valuation functions.
\end{informaltheorem}

Thus, in order to obtain lower bound results for our cake-cutting problem, it suffices to prove lower bounds for the Jordan curve problem. Our main technical contribution is to prove the first such lower bound result for the Jordan curve problem.

\begin{informaltheorem}\label{inf:JC}
The Jordan curve problem is \UEOPL-hard.
\end{informaltheorem}

Our main result in \cref{inf:CC} thus follows as a corollary\footnote{In particular, the query lower bound follows from the fact that the canonical \UEOPL-complete problem admits an exponential lower bound and our reductions are black-box.} of \cref{inf:JC2CC} and \cref{inf:JC}. Our current understanding of the complexity of these problems is illustrated in \cref{figure:class_diagram}.

The class \UEOPL, which stands for \emph{Unique End of Potential
  Line}, is a subclass of \PPAD which contains notoriously challenging
problems such as solving Simple Stochastic Games and the
P\nobreakdash-matrix Linear Complementarity
Problem~\cite{Fearnley21}. These problems have evaded attempts at
polynomial\nobreakdash-time algorithms for decades. Our
\UEOPL-hardness results for the cake cutting and Jordan curve problems
thus provide compelling evidence that the problems are unlikely to be
solvable in polynomial time, barring some truly breakthrough
algorithmic result. Further evidence of the intractability of
\UEOPL-hard problems is provided by a series of works which prove that
the class is hard under various cryptographic
assumptions~\cite{HubacekY2017-CLS,ChoudhuriHKPRR19-Fiat-Shamir,JawaleKKZ21-PPAD-LWE}.\footnote{These
  works usually state their results for the more widely recognizable
  class \PPAD, but they hold for a problem called ``Sink of Verifiable
  Line'', which lies in \UEOPL.}

Although our results already establish intractability for our two
problems of interest, there remains
a gap between the \UEOPL-hardness and \PPAD-membership. Understanding
the precise complexity of these problems is a major challenge. Even
determining whether the Jordan curve problem lies in \PLS or showing a
black-box separation with \PLS is a highly non-trivial
question. Finally, let us mention the intriguing possibility that the
problems might not be complete for any of the existing classes, but
for some new class or classes yet to be defined.

\begin{figure}
    \centering
    \includegraphics[scale=1]{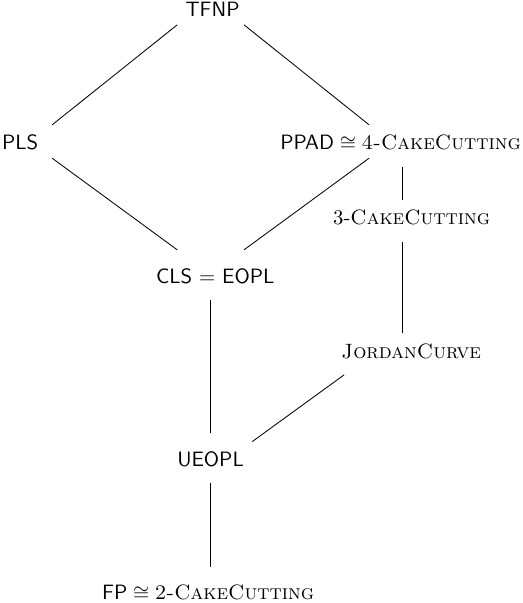}
    \caption{Class diagram. Here $k$-\CakeCutting refers to the cake-cutting problem with $k$ agents.}
    \label{figure:class_diagram}
\end{figure}

\section{Preliminaries}

\subsection{Cake cutting}

In cake-cutting, there is a divisible resource, called ``the cake'', which is represented as the interval $[0,1]$. Furthermore, there are $n$ agents, each with its own valuation function assigning a value to each piece (i.e., interval) of the cake. More formally, agent $i \in [n]$ has a valuation function $v_i: [0,1]^2 \to [0,1]$, where $v_i(a,b)$ represents the value that agent $i$ has for interval $[a,b]$. We require that $v_i(a,b) = 0$ whenever $b \leq a$, i.e., an empty piece of cake should have value zero.

The valuations are always assumed to be continuous, and furthermore, for computational purposes, we will assume that they are Lipschitz-continuous. A valuation $v$ is Lipschitz-continuous with Lipschitz constant $L$, if, for all $a,b,a',b' \in [0,1]$ with $a \leq b$ and $a' \leq b'$:
$$|v(a,b) - v(a',b')| \leq L (|a-a'| + |b-b'|).$$

We are interested in dividing the cake into $n$ pieces and assigning each piece to one of the $n$ agents. Formally, we wish to partition the cake $[0,1]$ into $n$ intervals $A_1, \dots, A_n$ such that $A_i$ is assigned to agent $i$.

\paragraph*{\textbf{Envy-freeness.}}
An allocation $(A_1, \dots, A_n)$ is \emph{envy-free} if, for all agents $i$, we have $v_i(A_i) \geq v_i(A_j)$ for all $j$. For $\varepsilon \in [0,1]$, we say that the allocation is \emph{$\varepsilon$-envy-free} if, for all agents $i$, we have $v_i(A_i) \geq v_i(A_j) - \varepsilon$ for all $j$. It is known that an envy-free allocation always exists when the agents are \emph{hungry}, i.e., they (weakly) prefer any non-empty piece to an empty piece~\cite{Stromquist80-existence,Woodall80-existence,Su99-rental-harmony}. Since we have assumed that valuations $v_i$ are always nonnegative and furthermore assign value zero to empty pieces, the hungriness condition is automatically satisfied and an envy-free allocation is guaranteed to exist.

\paragraph*{\textbf{Normalization.}} Without loss of generality, we
assume that the valuation functions are $1$-Lipschitz-continuous,
i.e., $L=1$. If $L > 1$, we can replace $v_i$ by $v_i/L$, and
$\varepsilon$ by $\varepsilon/L$. This normalization allows us to
express bounds on the number of queries in terms of $\eps$ only.

\paragraph*{\textbf{Query complexity.}}
In the query model, we can query the valuation functions of the agents. A query consists of the endpoints of an interval $[x,y]$, and the agent responds with its value for that interval, i.e., $v_i(x,y)$. The running time of an algorithm is the number of queries it makes. It is known that no finite algorithm exists for finding \emph{exact} solutions (i.e., with $\eps = 0$)~\cite{stromquist2008envy}.\footnote{This impossibility result holds even if one also allows the use of so-called \emph{cut queries}, which are commonly included in the Robertson-Webb model~\cite{robertson1998cake,WoegingerS07-cake}. For finding $\eps$-envy-free allocations, it is known that cut queries can be simulated by a logarithmic number of standard evaluation queries, and thus we do not need to include them in our model.}

This motivates the study of algorithms for finding $\eps$-envy-free allocations where we allow the number of queries to depend on $\eps > 0$. For a constant number of agents, the problem of computing an $\eps$-envy-free allocation can be solved using $\poly(1/\eps)$ queries by brute force~\cite{BranzeiN22-cake-query}. We say that an algorithm is efficient if it uses $\poly(\log(1/\eps))$ queries instead.

\paragraph*{\textbf{Computational complexity.}}
In the standard Turing machine model, the valuations are given to us in the input. An algorithm is efficient, if it runs in polynomial time in the size of the representation of the valuations and in $\log(1/\eps)$. For example, the valuations can be given as well-behaved arithmetic circuits \cite{Fearnley22} or as Turing machines (together with a polynomial upper bound on their running time). In this model, the problem is a total \NP search problem, i.e., it lies in the class \TFNP. Furthermore, it is known to lie in the subclass \PPAD of \TFNP \cite{DQS12}.

\subsection{The $\JordanCurve$ problem}
\label{sec:jordan-curve}

The input to a $\JordanCurve$ instance is a pair of circuits
$R, B: \ZO^{m} \to \ZO^{2d}$ which are interpreted as defining two
curves with timescale $[M]$ moving on a grid with dimension
$D \times D$ (where $D \coloneqq 2^d$ and $M \coloneqq 2^m$). More
precisely, $R$ represents a red curve, where for each timestamp
$t \in [M]$ (which is naturally identified with $\ZO^m$), the position
of the curve is given as a coordinate $(x, y)$ with
$x,y \in \llbracket 0, D-1 \rrbracket \coloneqq \{0,\, \dots, \,
D-1\}$ (identified with $\ZO^{d}$). The blue curve $B$ is defined
similarly. We take as convention that the coordinate $(0, 0)$
corresponds to the bottom-left corner and $(D-1, D-1)$ to the
top-right one. If the red and blue curves start and end in opposite
corners and are continuous, a solution is any pair of timestamps that
witness an intersection of the blue and the red curve. More precisely,
the following are all acceptable solutions:
\begin{enumerate}
\item $(i, j)$ if $R(i) = B(j)$ \itemcomment{red-blue crossing}%
\item $1$ if $R(1) \neq (0, 0)$ or $R(M) \neq (D-1, D-1)$
  \itemcomment{red wrong start/end}\label{jc:2}%
\item $1$ if $B(1) \neq (0, D-1)$ or $B(M) \neq (D-1, 0)$
  \itemcomment{blue wrong start/end}%
\item $i$ if $\|R(i) - R(i+1)\|_1 > 1$ \itemcomment{red discontinuous}%
\item $j$ if $\|B(j) - B(j+1)\|_1 > 1$ \itemcomment{blue
    discontinuous}\label{jc:5}%
\end{enumerate}
Let us now state a couple of observations about allowable range of
parameters and solutions types.
\paragraph{Range of parameters.} To make the $\JordanCurve$ problem
interesting, one needs~$M$ to be somewhat larger than~$2D$. Indeed, in
the extreme case of $M = 2D$, the red and blue curve are constrained
to be monotone and a solution is hence easy to find \cite{ADD16}. The
reduction of \Cref{theorem:ueopl_leq_jc} shows that it is \UEOPL-hard
to solve a $\JordanCurve$ instance in the regime~$M = D^C$ for
constant $C > 1$. We observe however that this readily implies
hardness for the parameter range~$M = 2D + D^\eps$ for any
constant~$\eps > 0$.

Indeed, let~$D' = D^{\eps/C}$ and consider a grid of size~$D' \times D'$ placed at the
center of the~$D \times D$ grid. Define now the red and blue curves on the $D \times D$ grid as first moving diagonally from their respective corners to the corresponding corners of the inner $D' \times D'$ grid. This requires $2 \cdot (D - D')$ timestamps per curve and as such there remain $M'$ timestamps to move within the inner $D' \times D'$ grid with
\[
M' = 2D + D^\eps - 2(D - D') = 2D' + D^\eps \geq (D')^{C}.
\]
Thus, the hardness of \cref{theorem:ueopl_leq_jc} applies to the inner grid and we may conclude that $\JordanCurve$ is \UEOPL-hard in the aforementioned regime already.
\paragraph{Enforcing no syntactic solutions.} We may assume without
loss of generality that $\JordanCurve$ instances~$(R, B)$ have solution of the first kind (red-blue crossings) only. Indeed, if~$(R,B)$ is defined over a grid of
dimension~$D \times D$ and timescale~$[M]$, we can get rid of the syntactic solutions (\ref{jc:2})--(\ref{jc:5}) by constructing the instance~$(R', B')$ as follows. $R'$ and $B'$ are defined on the same $D \times D$ grid, but with increased time-domain~$M' \coloneqq 2M \cdot D$. The curves are essentially copies of~$R$ and $B$ on multiples of $2D$:
\[
  R'(t') =
  \begin{cases}
    (0, 0)
    &\quad\text{if $t' = 1$,}\\
    (D-1, D-1)
    &\quad\text{if $t' = M'$,}\\
    R(t)
    &\quad\text{if $t' = 2t \cdot D$}
  \end{cases}
    \qquad B'(t') =
  \begin{cases}
    (0, D-1)
    &\quad\text{if $t' = 1$,}\\
    (D-1, 0)
    &\quad\text{if $t' = M'$,}\\
    R(t)
    &\quad\text{if $t' = 2t \cdot D$}
  \end{cases}
\]
The positions of the curves for the other timestamps are defined by interpolating (by a shortest path) between the fixed points above. For instance, for
timestamps~$t \in \llbracket 2t \cdot D, (2t+1) \cdot D\rrbracket$ the red curve~$R'$
interpolates by a shortest path from~$R(2t \cdot D)$
to~$R\big(2 \cdot (t+1) \cdot D\big)$. Note that a solution~$(t_1', t_2')$ of~$(R', B')$ is readily
mapped back to~$(R, B)$ by computing the corresponding
timestamps~$t_1, t_2$, to then either output a syntactic solution (in
case the curves~$B,R$ are discontinuous or have a wrong start/end), or
otherwise output the intersection~$(t_1, t_2)$. See
\cref{figure:syntactic_jc} for an illustration.

Finally let us remark that the above construction can also be used to
show that the \textsc{CrossingCurves} problem, as defined in
\cite{ADD16}, is equivalent to the \textsc{JordanCurve} problem.

\begin{figure}
    \centering
    \includegraphics[scale=1]{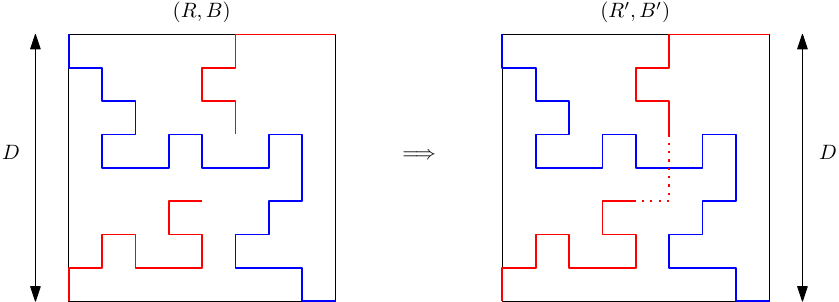}
    \caption{Completing a \JordanCurve instance to remove syntactic
      solutions.}
    \label{figure:syntactic_jc}
\end{figure}

\subsection{Unique end of potential line}

The class \UEOPL is defined by its canonical complete problem, the \textsc{Unique-End-of-Potential-Line} ($\ueopl$) problem.
We use the grid-like definition of $\ueopl$
\cite{Fearnley21,Goos24_separations}. An instance consists of two
circuits~$S, P\colon \ZO^{2n} \to \ZO^n \cup \{\bot\}$, termed the
\emph{successor} and the \emph{predecessor} circuit, that are
interpreted as a series of paths on the~$N \times N$ grid
for~$N \coloneqq 2^n$. The output of the successor
circuit~$S(i, j) = j_s$ is interpreted as node~$(i,j)$ having
successor~$(i + 1, j_s)$ while the output of the
predecessor~$P(i,j) = j_p$ is interpreted as node~$(i,j)$ having
predecessor~$(i-1,j_p)$. For~$i \neq N$, we say that node~$(i, j)$ is
\emph{active} if~$S(i, j) = j_s \neq \bot$ and~$P(i+1, j_s) = j$. We
further say that a node~$(N, j)$ is active if~$S(N, j) \neq \bot$. A
solution to the $\ueopl$ instance~$(S,P)$ is either an end of a line
or two parallel lines, that is,
\begin{enumerate}
    \item $(1, 1)$ if it is inactive
    \item $(N, j)$ if it is active \itemcomment{sink}%
    \item $(i, j)$ if it has an active predecessor but is inactive
      \itemcomment{proper sink}%
    \item $(i, j) \neq (1,1)$ if it has no predecessor but is active
      \itemcomment{proper source}%
    \item $(i, j)$, $(i, j')$ if both are active and $j \neq j'$
      \itemcomment{parallel lines}%
\end{enumerate}

\paragraph{Enforcing no syntactic solutions.} Following \cite{Goos24_collapse}, we may assume without
loss of generality that the predecessor and the successor circuits are
consistent, that~$S(N,j) = \bot$ for all~$j \in [N]$, and that the
node~$(1, 1)$ is active.

\section{A Connection Between $\JordanCurve$ and $\CakeCutting$}

Our main result is the following theorem.

\begin{theorem}\label{thm:cake-hard}
It is \UEOPL-hard to compute a connected $\eps$-envy-free division for three agents with identical (non-monotone) valuations, where $\eps$ is given in binary. Furthermore, in the query model, the problem requires $\poly(1/\eps)$ queries.
\end{theorem}

This theorem follows from \cref{thm:JC2CC} stated and proved below, together with the black-box reduction from \ueopl to \JordanCurve provided in the next section (\cref{theorem:ueopl_leq_jc}). See \cref{rem:query-lower-bound} at the end of the proof for more details.

\begin{theorem}\label{thm:JC2CC}
There is a polynomial-time black-box reduction from $\JordanCurve$ to $\CakeCutting$ with three identical agents.
\end{theorem}

We proceed with the proof of \cref{thm:JC2CC}. In this section it is
most convenient to work with a continuous version of the
$\JordanCurve$ problem rather than the combinatorial description as
provided in \cref{sec:jordan-curve}. It is readily seen to be
equivalent to said description by simply connecting the discrete
subsequent points by straight lines.

Let $R, B: [0,1] \to [-1,1]^2$ be an instance of the continuous
$\JordanCurve$ problem, where $R$ and $B$ are $L$-Lipschitz-continuous
(with respect to the $\ell_\infty$-norm) for some given $L \geq
10$. We are also given $\eps' \in (0,1)$ and the goal is to output two
timestamps~$t_1, t_2 \in [0,1]$ such that the two curves
are~$\eps'$-close on these timestamps, that is,
$\|R(t_1) - B(t_2)\|_\infty \leq \eps'$.

\paragraph{High-level idea.}
Given the curves $R$ and $B$, our goal will be to construct a continuous valuation function $v$ over the cake $[0,1]$ that satisfies the following two desiderata:
\begin{enumerate}
    \item In any approximately envy-free division $(c_1,c_2)$ of the cake, it must be that the first cut $c_1$ lies in a small interval around $1/3$, and the second cut $c_2$ lies in a small interval around $2/3$.
    \item When the cuts $(c_1,c_2)$ lie close to $(1/3,2/3)$, the value of the first piece $v(0,c_1)$ is approximately equal to $R_1(c_1)+R_2(c_1)$, where $R_1$ and $R_2$ are the first and second coordinates of the curve $R$, respectively. Similarly, we want the value of the third piece $v(c_2,1)$ to be approximately equal to $B_1(c_2)+B_2(c_2)$. Finally, we want the value of the middle piece $v(c_1,c_2)$ to be approximately equal to $R_1(c_1)+B_2(c_2)$.
\end{enumerate}
At an approximately envy-free division the values of the three intervals must be very close to each other (because all three agents share the same valuation function $v$). As a result, we obtain that
$$R_1(c_1)+R_2(c_1) \approx R_1(c_1)+B_2(c_2) \approx B_1(c_2)+B_2(c_2)$$
which implies that
$$R_1(c_1) \approx B_1(c_2) \quad \text{ and } \quad R_2(c_1) \approx B_2(c_2)$$
i.e., $(c_1,c_2)$ yields a solution to the Jordan curve problem, as desired.

The first desiderata can be achieved by defining $v(a,b)$ to be equal to $b-a$ for any $a \leq b$. This enforces the first desiderata, but makes it trivial to find a solution, namely $(c_1,c_2) = (1/3,2/3)$. Thus, instead we only enforce that $v(a,b)$ is sufficiently close to $b-a$. This still ensures that a solution must lie close to $(1/3,2/3)$, but now we can ``hide'' its precise location using the Jordan curve problem. This requires some very careful modifications on the two curves so that the second desiderata can be achieved, but without introducing unwanted solutions when cuts move from the ``local'' regime around $1/3$ or $2/3$ to the ``global'' regime away from these two points. Indeed, we have to ensure that the valuation function $v$ is continuous everywhere and the challenge is to enforce both desiderata under this constraint. We now proceed with the formal presentation of the reduction and the proof of its correctness.

\paragraph{Modification of the $\JordanCurve$ instance.}
We begin by defining a modified Jordan curve instance. Let $m = 1000$. The modified red curve $\widehat{R}: [2/9,4/9] \to [-1/18,7/18]^2$ is defined as follows:
\begin{itemize}
\item for $t \in [2/9,3/9-2/m] \cup [3/9+2/m,4/9]$, we let $\widehat{R}(t) = (1/2 - t, -1/2 + 2t)$,
\item for $t \in [3/9-2/m, 3/9-1/m]$, we let $\widehat{R}(t)$ be the straight line going from $(1/6 + 2/m, 1/6 - 4/m)$ to $(1/6 - 1/m,1/6 - 1/m)$,
\item for $t \in [3/9+1/m, 3/9+2/m]$, we let $\widehat{R}(t)$ be the straight line going from $(1/6 + 1/m, 1/6 + 1/m)$ to $(1/6 - 2/m,1/6 + 4/m)$,
\item for $t \in [3/9 - 1/m, 3/9 + 1/m]$, we let $\widehat{R}(t) = (1/6, 1/6) + (1/m) \cdot R((t-3/9+1/m) \cdot (m/2))$.
\end{itemize}
Similarly, the modified blue curve $\widehat{B}: [5/9,7/9] \to [-1/18,7/18]^2$ is defined as follows:
\begin{itemize}
\item for $t \in [5/9,6/9-2/m] \cup [6/9+2/m,7/9]$, we let $\widehat{B}(t) = (3/2 - 2t, -1/2 + t)$,
\item for $t \in [6/9-2/m, 6/9-1/m]$, we let $\widehat{B}(t)$ be the straight line going from $(1/6 + 4/m, 1/6 - 2/m)$ to $(1/6 + 1/m,1/6 - 1/m)$,
\item for $t \in [6/9+1/m, 6/9+2/m]$, we let $\widehat{B}(t)$ be the straight line going from $(1/6 - 1/m, 1/6 + 1/m)$ to $(1/6 - 4/m,1/6 + 2/m)$,
\item for $t \in [6/9 - 1/m, 6/9 + 1/m]$, we let $\widehat{B}(t) = (1/6, 1/6) + (1/m) \cdot B((6/9 + 1/m - t) \cdot (m/2))$.
\end{itemize}
The high level idea of this embedding is illustrated in
\cref{fig:cake-high-level}. The following claim easily follows from
the construction.

\begin{figure}
  \centering
  \includegraphics{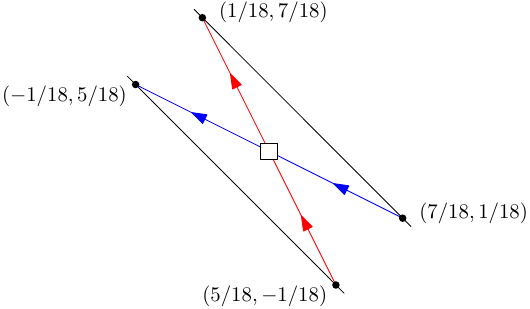}
  \caption{The modified $\JordanCurve$ instance with arrows indicating
    the passing of time}
  \label{fig:cake-high-level}
\end{figure}

\begin{claim}\label{clm:modif-properties}
The modified instance $\widehat{R}, \widehat{B}$ has the following properties:
\begin{enumerate}
\item $\widehat{R}, \widehat{B}$ are $L/2$-Lipschitz-continuous,
\item for all $t \in [2/9,4/9]$, we have $\|\widehat{R}(t) - (1/2 - t,-1/2 + 2t)\|_\infty \leq 8/m$,
\item for all $t \in [5/9,7/9]$, we have $\|\widehat{B}(t) - (3/2 - 2t, -1/2 + t)\|_\infty \leq 8/m$,
\item if $\|\widehat{R}(t_1) - \widehat{B}(t_2)\|_\infty \leq \eps'/2m$, then $t_1' := \mathsf{T}[(t_1-3/9+1/m) \cdot (m/2)]$ and $t_2' := \mathsf{T}[(6/9 + 1/m - t_2) \cdot (m/2)]$ satisfy $\|R(t_1') - B(t_2')\|_\infty \leq \eps'$,
\end{enumerate}
where $\mathsf{T}: \mathbb{R} \to [0,1]$ denotes truncation to the $[0,1]$ interval, i.e., $\mathsf{T}[x] = \min\{1, \max\{0, x\}\}$.
\end{claim}

\paragraph{Construction of the valuation function.}
We construct three agents who all have the same valuation function $v$ over the cake. Let $N = \lceil 8mL/\eps' \rceil$ and $D$ be a $1/N$-fine discretization of $[0,1]$, i.e., subsequent points in $D$ are at distance $1/N$ of each other. For $i=1,2$, we let $\widehat{R}_i(t)$ denote the $i$th component of the curve, i.e., $\widehat{R}(t) = (\widehat{R}_1(t), \widehat{R}_2(t))$, and similarly $\widehat{B}(t) = (\widehat{B}_1(t), \widehat{B}_2(t))$.

We begin by defining $v$ for points lying on the discrete grid $D$. Namely,
\begin{itemize}
\item if $a=0$ and $b \in D \cap [2/9,4/9]$, then $v(0,b) = \widehat{R}_1(b) + \widehat{R}_2(b)$,
\item if $a \in D \cap [5/9,7/9]$ and $b = 1$, then $v(a,1) = \widehat{B}_1(a) + \widehat{B}_2(a)$,
\item if $a \in D \cap [2/9,4/9]$ and $b \in D \cap [5/9,7/9]$, then $v(a,b) = \widehat{R}_1(a) + \widehat{B}_2(b)$,
\item otherwise, for any other $a, b \in D$ with $a \leq b$, let $v(a,b) = b-a$.
\end{itemize}
Then, we obtain a continuous valuation function defined for all
$a,b \in [0,1]$ with $a \leq b$ by taking the piecewise linear
interpolation (see,
e.g.,~\cite[Section~5.1]{HollenderR23-cake-four}). In particular, note
that $v$ is $O(N)$-Lipschitz-continuous. At a high level, the
valuation function encodes information from the $\JordanCurve$
instance when the cuts lie in the special regions illustrated in
\cref{fig:interval}.

\paragraph{Correctness.}
The following claim says that the valuation function always outputs a value that is relatively close to the length of the interval.

\begin{claim}\label{clm:valuations-approx-length}
For any $a,b \in [0,1]$ with $a \leq b$ we have $|v(a,b) - (b-a)| \leq 16/m$.
\end{claim}

\begin{proof}
Since $v$ is obtained by piecewise linear interpolation over the grid $D$, and since $(a,b) \mapsto b-a$ is a linear function, it suffices to show that the statement holds for all $a,b \in D$ with $a \leq b$.
\begin{itemize}
\item If $a=0$ and $b \in D \cap [2/9,4/9]$, then $v(0,b) = \widehat{R}_1(b) + \widehat{R}_2(b)$, and by property 2 of \cref{clm:modif-properties}, we have $|\widehat{R}_1(b) + \widehat{R}_2(b) - b| \leq 16/m$, which implies $|v(0,b) - (b-0)| \leq 16/m$.
\item If $a \in D \cap [5/9,7/9]$ and $b = 1$, then $v(a,1) = \widehat{B}_1(a) + \widehat{B}_2(a)$, and by property 3 of \cref{clm:modif-properties}, we have $|\widehat{B}_1(a) + \widehat{B}_2(a) - (1-a)| \leq 16/m$, and thus $|v(a,1) - (1-a)| \leq 16/m$.
\item If $a \in D \cap [2/9,4/9]$ and $b \in D \cap [5/9,7/9]$, then $v(a,b) = \widehat{R}_1(a) + \widehat{B}_2(b)$, and by properties 2 and 3 of \cref{clm:modif-properties}, we have $|\widehat{R}_1(a) + \widehat{B}_2(b) - (b-a)| \leq 16/m$, i.e., $|v(a,b) - (b-a)| \leq 16/m$, as desired.
\item Finally, for any other $a, b \in D$ with $a \leq b$, we have $v(a,b) = b-a$, so the statement holds trivially.
\end{itemize}
We have thus shown that the statement holds for all $a,b$ on the grid, and it follows that it holds for all $a,b \in [0,1]$ with $a \leq b$.
\end{proof}

\begin{figure}
  \centering
  \includegraphics{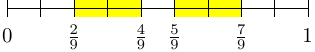}
  \caption{Highlighted are the two regions of the valuation function
    that encode information from the~$\JordanCurve$ instance}
  \label{fig:interval}
\end{figure}

Set $\eps = \eps'/4m$. The following claim states that at any $\eps$-envy-free division the cuts have to lie close to $(1/3,2/3)$.

\begin{claim}\label{clm:cuts-in-regions}
Any $\eps$-envy-free division $(c_1,c_2)$ satisfies $c_1 \in [1/3 - 23/m, 1/3 + 23/m]$ and $c_2 \in [2/3 - 23/m, 2/3 + 23/m]$.
\end{claim}

\begin{proof}
Assume towards a contradiction that $c_1 \leq 1/3 - 23/m$. Then, by \cref{clm:valuations-approx-length} we obtain that $v(0,c_1) \leq 1/3 - 23/m + 16/m = 1/3 - 7/m$. Since $(c_1,c_2)$ is $\eps$-envy-free, it follows that $v(c_2,1) \leq v(0,c_1) + \eps \leq 1/3 - 7/m + 1/m = 1/3 - 6/m$, where we used $\eps \leq 1/m$. By applying \cref{clm:valuations-approx-length} again, we obtain that $1 - c_2 \leq 1/3 - 6/m + 16/m$, i.e., $c_2 \geq 2/3 - 10/m$. Thus, we have that $c_2 - c_1 \geq 1/3 - 10/m + 23/m = 1/3 + 13/m$, which by \cref{clm:valuations-approx-length} implies that $v(c_1,c_2) \geq 1/3 + 13/m - 16/m = 1/3 - 3/m > 1/3 - 7/m + \eps \geq v(0,c_1) + \eps$, where we used $\eps \leq 1/m$ again. This is a contradiction to the fact that $(c_1,c_2)$ is an $\eps$-envy-free division.

The other bounds can be proved in a completely analogous manner.
\end{proof}

Finally, the following lemma states that we can recover a solution of the original problem we reduced from.

\begin{lemma}
Any $\eps$-envy-free division $(c_1,c_2)$ yields a solution to the original $\JordanCurve$ instance.
\end{lemma}

\begin{proof}
We begin with the following observation. The functions $v(0,a)$, $v(a,b)$, and $v(b,1)$ are $L$-Lipschitz-continuous when restricted to $(a,b) \in [1/3 - 50/m, 1/3 + 50/m] \times [2/3 - 50/m, 2/3 + 50/m]$. Indeed, since $m = 1000$, we have that $[1/3 - 50/m, 1/3 + 50/m] \subseteq [2/9, 4/9]$, and thus $v(0,a)$ is just a linear interpolation of the $L$-Lipschitz-continuous function $a \mapsto \widehat{R}_1(a) + \widehat{R}_2(a)$, which implies that it is itself also $L$-Lipschitz-continuous. Analogous arguments show that $(a,b) \mapsto v(a,b)$ and $b \mapsto v(b,1)$ are also $L$-Lipschitz-continuous over $(a,b) \in [1/3 - 50/m, 1/3 + 50/m] \times [2/3 - 50/m, 2/3 + 50/m]$.

Now, by \cref{clm:cuts-in-regions} we know that $(c_1,c_2)$ lie well within $[1/3 - 50/m, 1/3 + 50/m] \times [2/3 - 50/m, 2/3 + 50/m]$. Let $t_1 \in D$ denote the closest grid point to $c_1$, and $t_2 \in D$ denote the closest grid point to $c_2$ (break ties arbitrarily). In particular, $|c_1-t_1| \leq 1/N$ and $|c_2-t_2| \leq 1/N$. Then, by the $L$-Lipschitz continuity of $v$ mentioned in the previous paragraph, it follows that the division $(t_1,t_2)$ is $(\eps + 2L/N)$-envy-free.

Since $t_1 \in D \cap [2/9,4/9]$ and $t_2 \in D \cap [5/9,7/9]$, we have that $v(0,t_1) = \widehat{R}_1(t_1) + \widehat{R}_2(t_1)$ and $v(t_1,t_2) = \widehat{R}_1(t_1) + \widehat{B}_2(t_2)$. As a result, the fact that $|v(0,t_1) - v(t_1,t_2)| \leq \eps + 2L/N$ implies $|\widehat{R}_2(t_1) - \widehat{B}_2(t_2)| \leq \eps + 2L/N$. Similarly, since we also have $v(t_2,1) = \widehat{B}_1(t_2) + \widehat{B}_2(t_2)$, the fact that $|v(t_2,1) - v(t_1,t_2)| \leq \eps + 2L/N$ implies $|\widehat{R}_1(t_1) - \widehat{B}_1(t_2)| \leq \eps + 2L/N$. Finally, since $\eps + 2L/N \leq \eps'/4m + \eps'/4m = \eps'/2m$, by property 4 of \cref{clm:modif-properties}, we obtain an $\eps'$-approximate solution to the original $\JordanCurve$ instance.
\end{proof}

\begin{remark}\label{rem:query-lower-bound}
The reduction in \cref{thm:JC2CC} yields \cref{thm:cake-hard} as follows. The \UEOPL-hardness follows immediately from \cref{theorem:ueopl_leq_jc}. For the query lower bound, \cref{theorem:ueopl_leq_jc} implies that we need $\poly(D)$ queries to solve the discrete version of $\JordanCurve$ over a $D \times D$ grid with timescale $M = O(D)$. Continuous interpolation of the discrete instance yields a continuous instance where the curves $R, B: [0,1] \to [-1,1]^2$ are $O(1)$-Lipschitz-continuous and where we are looking for timestamps where the two curves are at distance at most $\varepsilon' = \Theta(1/D)$ of each other. This implies a $\poly(1/\varepsilon')$ query lower bound for the continuous version of $\JordanCurve$ with a $O(1)$ Lipschitz constant. Finally, the reduction in \cref{thm:JC2CC} produces a cake-cutting instance with $\varepsilon = \Theta(\varepsilon')$ and where the valuation function is $O(1/\varepsilon')$-Lipschitz-continuous. After normalization of the Lipschitz constant to $1$, we have $\varepsilon = \Theta(\varepsilon'^2)$. Thus, we obtain a query lower bound of $\poly(1/\varepsilon)$ for cake-cutting.
\end{remark}

\section{A Lower Bound for the $\JordanCurve$ Problem}

In this section we show how to reduce a \ueopl instance to an instance of the Jordan Curve problem. In other words we show that
the Jordan Curve problem is $\UEOPL$-hard.

\begin{theorem}\label{theorem:ueopl_leq_jc}
The $\JordanCurve$ problem over a $D \times D$ grid is $\UEOPL$-hard even when the timescale $M$ is as small as $M = 2D + D^\varepsilon$ for any constant $\varepsilon > 0$.
\end{theorem}

Note that since $\ueopl$ on an $N \times N$ grid has a query complexity lower bound of $\poly(N)$ and the reduction in the proof of \cref{theorem:ueopl_leq_jc} is black-box, it follows that $\JordanCurve$ over a $D \times D$ grid has a query complexity lower bound of $\poly(D)$.

We prove \cref{theorem:ueopl_leq_jc} by a reduction from \ueopl to
\JordanCurve. In \cref{subsection:desc_ueopl_jc} we describe the
reduction while we prove its correctness in
\cref{subsection:correctness_ueopl_jc}.

\subsection{Description of the Reduction}
\label{subsection:desc_ueopl_jc}

Fix a $\ueopl$ instance $(S, P)$ defined on a grid of size
$N \times N$ and assume without loss of generality that it has no
syntactic solutions (if not, we may first apply the transformation
that gets rid of those solutions and then reduce to
$\JordanCurve$). The resulting $\JordanCurve$ instance~$(B,R)$ will be
defined over a grid of size~$D \times D$ for~$D \coloneqq 3N^2 + 1$
and a time-domain $\cT$ of length $M \coloneqq (N^2 + 2) \cdot 4D$. Recall from the discussion on parameters in \cref{sec:jordan-curve} that this implies hardness for the regime $M = 2D + D^\eps$ for any constant $\eps > 0$. The hardness claimed in \cref{theorem:ueopl_leq_jc} thus follows.

\newcommand{\calT}{\mathcal{T}}

\paragraph{High-level idea.} The red curve~$R$ (respectively, the blue
curve~$B$) almost always lies on row~$2D/3$ (resp., on row~$D/3$) of
the~$D\times D$ grid. Associate each node~$(i,j)$ of the \ueopl
instance with odd~$i$ with a small time interval~$\calT_{i,j}$ of the
red curve and those with even~$i$ with a small time
interval~$\calT_{i,j}$ of the blue curve such that
if~$(i,j) < (i',j')$ lexicographically (i.e., $i < i'$, or $i=i'$ and
$j < j'$), then~$\calT_{i,j}$ comes before $\calT_{i',j'}$. Furthermore, for each $(i,j)$ define a location $C^R_{i,j}$ on the default row $2D/3$ such that the locations are evenly spaced and lexicographically ordered (see \cref{figure:checkpoints}) and analogously define locations $C^B_{i,j}$ on the default row $D/3$.
For
even~$i$ and~$\calT_{i,j} = \llbracket t_0, \ldots, t_{\textrm{end}}\rrbracket$,
fix~$R(t_0) = C^R_{i,j}$ and~$R(t_{\textrm{end}}) = C^R_{i,j+1}$; similarly for odd~$i$ fix the
blue curve~$B(t_0) = C^B_{i,j}$
and~$B(t_{\textrm{end}}) = C^B_{i,j+1}$.

Unless a node~$(i,j)$ has a predecessor and a successor (let us ignore
the node~$(0,0)$ for this discussion), the corresponding curve remains
on the default row throughout~$\calT_{i,j}$. Consider a node~$(i,j)$
with a predecessor~$P(i,j)=(i_p,j_p)$ and a
successor~$S(i,j)=(i_s,j_s)$. Suppose that~$i$ is even so that~$(i,j)$
is associated with an interval~$\calT_{i,j}$ on the red curve~$R$. The
main idea is to let~$R$ pass beneath the blue curve~$B$
during~$\calT_{i,j}$: during this time interval the red curve starts
at~$C^R_{i,j}$ to first pass by~$C^B_{i_p,j_p+1}$ to then pass beneath
the blue curve to the point~$C^B_{i_s,j_s}$ and connect back
to~$C^R_{i,j+1}$. If a symmetric construction is applied to the
intervals on the blue curve~$B$ the blue curve is ``above'' the red
curve~$R$ even though on most time stamps it seems to lie beneath~$R$
(see \cref{figure:reduction} for an illustration). This construction
ensures that the red and blue curves only cross at timestamps
corresponding to solutions of the \ueopl instance. Thus, although it
is clear that the curves must cross each other at some point, locating
such a crossing is hard. In the following we make this intuition
rigorous.

\paragraph{Time and space checkpoints.} To streamline the description
of the reduction we introduce several special positions in the space
domain of the $\JordanCurve$ instance and partition its time-domain
into intervals.
Large parts of the red curve~$R$ will simply lie on
row~$y_\text{R} \coloneqq 2D/3$, whereas the blue curve~$B$ will
mostly follow row~$y_\text{B} \coloneqq D/3$.
We introduce so-called \emph{checkpoints} on the just mentioned
rows. The curves \emph{always} run through these positions independently
of the $\ueopl$ instance~$S, P$: for~$i,j \in [N]$ let
\begin{align*}
  \checkR_{i,j}
  &\coloneqq
    \big(
    3N \cdot (i-1) +
    3 \cdot (j - 1)
    ,
    \,
    y_\text{R}
    \big),
  &
    \checkR_{\text{end}}
  &\coloneqq
    (D-1,\, y_\text{R}),\\[1em]
  \checkB_{i,j}
  &\coloneqq
    \big(
    3N \cdot (i-1) +
    3 \cdot (j - 1),
    \,
    y_\text{B}
    \big),\text{~and}
  &
    \checkB_{\text{end}}
  &\coloneqq
    (D-1,\, y_\text{B}).
\end{align*}
We partition the time-domain~$\cT$ into~$N^2 + 2$ contiguous intervals
of length~$4D$ each. On the respective curve each time interval will
eventually be used to connect two consecutive checkpoints.
For~$i,j \in [N]$ the intervals are
defined by
\begin{align*}
  \cT_\text{start}
  &\coloneqq
    \big\llbracket0,
    4D - 1\big\rrbracket,\\[1em]
  \cT_{i, j}
  &\coloneqq
    \big\llbracket
    \big(
    (i- 1)\cdot N +j
    \big)
    \cdot
    4D,
    \big(
    (i - 1) \cdot N + j + 1
    \big)
    \cdot
    4D - 1
    \big\rrbracket,
    \text{~and}\\[1em]
  \cT_\text{end}
  &\coloneqq
    \big\llbracket
    (N^2 + 1)
    \cdot
    4D,
    (N^2 + 2)
    \cdot
    4D - 1
    \big\rrbracket.
\end{align*}
Finally, for ease of notation, let us denote
by~$\nextT\colon [N] \times [N] \rightarrow [N] \times [N] \cup
\{\text{end}\}$ the function defined by
\begin{align*}
  \nextT(i,j)
  &=
    \begin{cases}
      (i, j+1)
      &\text{if~$j < N$,}\\
      (i+1, 1)
      &\text{if~$i < N$ and~$j = N$, and}\\
      \text{end}
      &\text{if~$i,j = N$.}
    \end{cases}
\end{align*}

Recall that the red curve starts at position~$(0, 0)$ and finishes at
position~$(D-1, D-1)$, while the blue curve starts at~$(0,D-1)$ and
finishes at~$(D-1, 0)$. We guarantee that the curves start and end at
the correct positions and have no discontinuity by the following high
level plan summarized in \cref{figure:checkpoints}. The red curve
\begin{itemize}
\item remains throughout~$\cT_{\text{start}}$ at position~$(0, 0)$,
\item uses the interval~$\cT_{1,1}$ to move from~$(0,0)$ to
  the checkpoint~$\checkR_{1,2}$,
\item uses the other intervals~$\cT_{i,j}$ to move
  from~$\checkR_{i,j}$ to~$\checkR_{\nextT(i,j)}$, and
\item uses the final time interval~$\cT_\text{end}$ to move
  from~$\checkR_{\text{end}}$ to $(D-1, D-1)$.
\end{itemize}
Even though we promised to pass through every checkpoint the red curve
will in fact skip~$\checkR_{1,1}$. The curves will pass through all
other checkpoints, though. Similarly to the red curve, the blue curve
\begin{itemize}
\item uses the time interval~$\cT_{\text{start}}$ to move
  from~$(0,D-1)$ to~$\checkB_{1,1}$,
\item uses the intervals~$\cT_{i,j}$ to move from~$\checkB_{i,j}$
  to~$\checkB_{\nextT(i,j)}$, and
\item uses the final time interval~$\cT_\text{end}$ to move
  from~$\checkB_{\text{end}}$ to $(D-1, 0)$.
\end{itemize}

\begin{sidewaysfigure}[hp]
    \centering
    \includegraphics[scale=1]{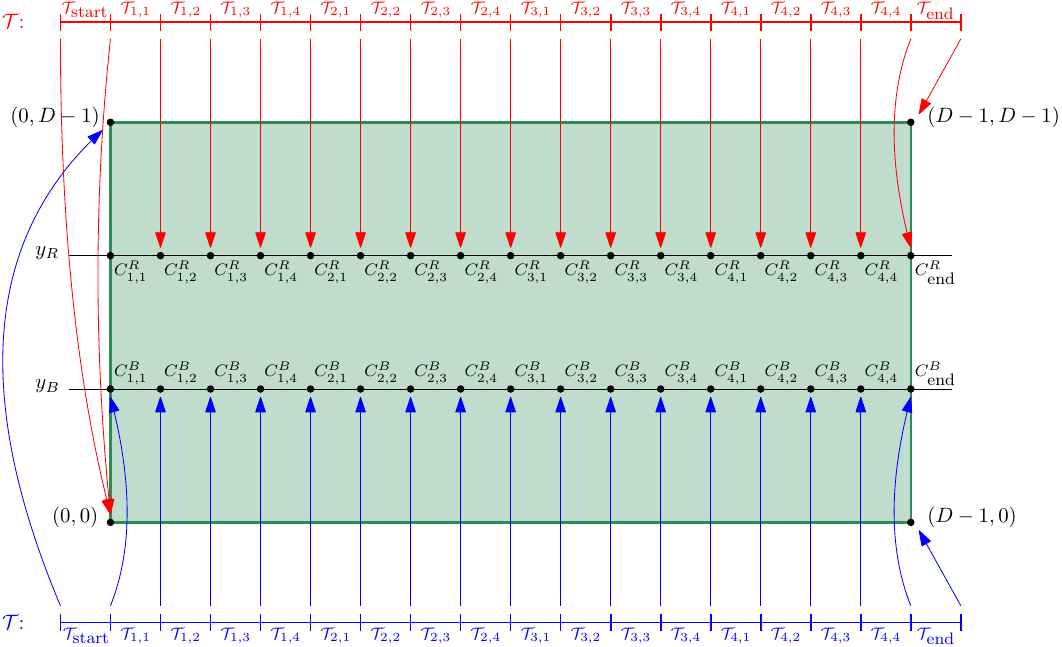}
    \vspace{1cm}
    \caption{Curve checkpoints in the reduction from a $\ueopl$
      instance of size~$4 \times 4$ to $\JordanCurve$}
    \label{figure:checkpoints}
\end{sidewaysfigure}

\paragraph{Definition of $\bm{R}$ and $\bm{B}$.} In the following we
define the two curves~$R$ and~$B$ in detail. It may be instructive to
refer to the example depicted in \cref{figure:reduction}.
Let us first define the curves~$R$ and~$B$ on the time intervals where
the curves are defined independent of the fixed~$\ueopl$
instance~$S,P$. On all these intervals the curves are straight lines
connecting two positions of the~$D \times D$ grid. Since every time
interval consists of~$4D$ points it is readily verified that any two
positions on the grid may be connected without introducing any
discontinuities.

During the time interval~$\cT_{\text{start}}$ the blue curve~$B$
connects the starting position~$(0, D-1)$ with the first
checkpoint~$\checkB_{1,1}$ by following the $0$-column. As previously
mentioned for any timestamp~$t \in \cT_{\text{start}}$ it holds
that~$R(t) = (0,0)$.  On time intervals~$\cT_{i,j}$ with even~$i$ the
red curve~$R$ connects the checkpoints~$\checkR_{i,j}$
to~$\checkR_{\nextT(i,j)}$ by following row~$y_R$. Similarly, on time
intervals~$\cT_{i,j}$ with~$i$ odd, the blue curve~$B$
connects~$\checkB_{i,j}$ with~$\checkB_{\nextT(i,j)}$ by following
row~$y_B$. On the final time interval~$\cT_{\text{end}}$ both curves
follow column~$D-1$ to either connect~$\checkB_{\text{end}}$
to~$(D-1,0)$ in case of the blue curve~$B$, or to
connect~$\checkR_{\text{end}}$ to~$(D-1, D-1)$ in case of the red
curve~$R$.

The remaining time intervals are more interesting since the curves
during these intervals depend on the $\ueopl$
instance~$(S,P)$. Consider some node~$(i,j)$ of the fixed~$\ueopl$
instance. Suppose~$i$ is odd -- the~$i$ even case is symmetric. The
basic idea is that if~$(i,j)$ has a predecessor~$P(i,j) = (i_p,j_p)$
as well as a successor~$S(i,j) = (i_s, j_s)$, then the red curve~$R$
will connect in the time interval~$\cT_{i,j}$ the
checkpoint~$\checkR_{i,j}$ to the position~$\checkB_{i_p,j_p} - (1,0)$
which is in turn connected with~$\checkB_{i_s,j_s} + (1,0)$ along
row~$y_B - 1$ to finally connect to the
checkpoint~$\checkR_{\nextT(i,j)}$. We stress that this implies the
red curve actually ventures on the blue curve baseline $y_B$.
If the symmetric construction is
used on the blue curve for the case when~$i$ is even, it is not so
hard to see that we can choose the curves such that no solution occurs
as long as there is no sink or a parallel path (see
\cref{figure:reduction}).

Let us define these paths in a bit more detail. Consider a
node~$(i,j) \neq (1,1)$ of the~$\ueopl$ instance with
odd~$i$. If~$(i,j)$ is inactive, has no predecessor, or no successor,
then~$R$ connects within the time interval~$\cT_{i,j}$ the
checkpoints~$\checkR_{i,j}$ and~$\checkR_{\nextT(i,j)}$ by following
row~$y_R$. Otherwise the node~$(i,j)$ is active, has a
successor~$S(i,j) = (i_s, j_s)$, and a
predecessor~$P(i,j) = (i_p, j_p)$. In this case the red curve~$R$
connects the checkpoint~$\checkR_{i,j}$ to the
position~$\checkR_{i,j} - (0,1)$. The curve then follows row~$y_R - 1$
to the left until it hits the column of
position~$\checkB_{i_p,j_p} - (1,1)$ to then connect along this column
with the just mentioned position. From there it goes to the right
along row~$y_B - 1$ to position~$\checkB_{i_s, j_s} + (1,-1)$. From
this position the curve follows the column upwards to row~$y_R -
1$. Once on row~$y_R -1$ it is followed to the left until
position~$\checkR_{\nextT(i,j)} - (0,1)$ to then end at the
checkpoint~$\checkR_{\nextT(i,j)}$. Since the described curve is
essentially a sub-curve of a rectangle, it is readily seen that this
path is of length at most~$4D$; the interval~$\cT_{i,j}$ suffices to
connect the mentioned positions of the grid without introducing
discontinuities.

The construction for the blue curve is symmetric: if a node~$(i,j)$
with~$i$ even is inactive, has no predecessor, or no successor,
then~$B$ connects in the time interval~$\cT_{i,j}$ the
checkpoint~$\checkB_{i,j}$ to~$\checkB_{\nextT(i,j)}$ along
row~$y_B$. Otherwise, that is, if~$(i,j)$ is active, has a
successor~$S(i,j) = (i_s, j_s)$, and a
predecessor~$P(i,j) = (i_p, j_p)$, then~$B$ connects the
checkpoint~$\checkB_{i,j}$ to position~$\checkB_{i,j} + (0,1)$ to then
follow row~$y_B + 1$ to the left until the column of
position~$\checkR_{i_p, j_p} + (-1, +1)$. Connect to the mentioned
position and follow row~$y_R+1$ to the right until
position~$\checkR_{i_s, j_s} + (1, 1)$. The curve follows the column
straight down to row~$y_B+1$ to then follow this row to the left until
it hits the column of the checkpoint~$\checkB_{\nextT(i,j)}$ to end at
the mentioned checkpoint.

It remains to define the behavior of the red curve~$R$ on the time
interval~$\cT_{1,1}$. Let~$(i_s, j_s) = S(1,1)$ be the successor
of~$(1,1)$. Note that this successor exists since we may assume that
the $\ueopl$ instance~$(S,T)$ has no syntactic solutions. The red
curve~$R$ starts at~$(0,0)$ to follow the $0$-row until it hits the
column of position~$\checkB_{i_s, j_s} + (1,0)$. It then follows this
column up to row~$y_R - 1$. The curve then follows this row until it
hits the column of the checkpoint~$\checkR_{1,2}$ to end at this
checkpoint.

\begin{sidewaysfigure}[hp]
  \centering
  \includegraphics[scale=1]{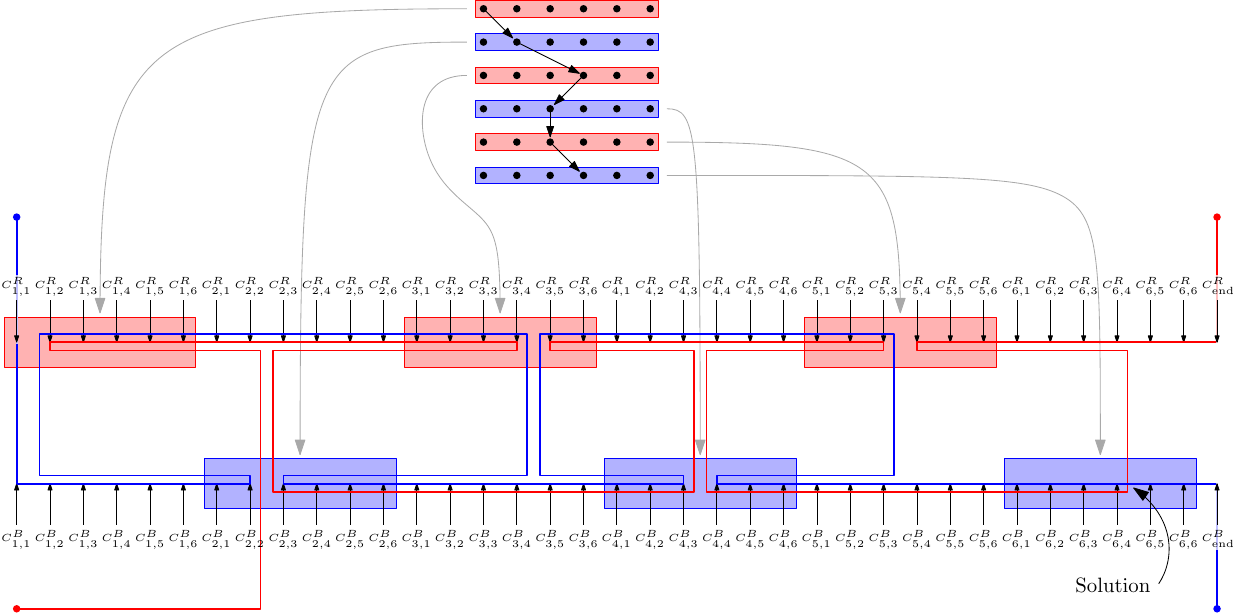}
  \vspace{1cm}
  \caption{An illustration of the reduction from $\ueopl$ to
    $\JordanCurve$}
  \label{figure:reduction}
\end{sidewaysfigure}

\subsection{Correctness of the Reduction}
\label{subsection:correctness_ueopl_jc}

Recall that the red curve is fixed on some parts irrespective of the
$\ueopl$ instance $S, P$. The same is true for the blue curve. Let us
call these time intervals \emph{dull} and denote them by
\begin{align*}
  \cT^R_d
  &\coloneqq
    \cT_{\text{start}} \cup \cT_{\text{end}} \cup
    \bigcup\nolimits_{i \,\text{even}}
    \bigcup\nolimits_{j \in [N]}
    \cT_{i,j}, \text{~and}\\[.2em]
  \cT^B_d
  &\coloneqq
    \cT_{\text{start}} \cup \cT_{\text{end}} \cup
    \bigcup\nolimits_{i \,\text{odd}}
    \bigcup\nolimits_{j \in [N]}
    \cT_{i,j}.
\end{align*}
We further call a time interval~$\cT_{i,j}$ \emph{interesting for~$R$}
if~$R$ does \emph{not} remain on row~$y_R$ throughout the time
interval; similarly we call an interval~$\cT_{i, j}$ \emph{interesting
  for~$B$} if~$B$ does not remain on row~$y_B$ during the
interval. Note that the~$\cT_{1,1}$ is always interesting for~$R$,
that dull intervals are never interesting for either curve, and that
there are intervals that are neither dull nor interesting for a curve:
intervals that correspond to inactive nodes of the~$\ueopl$
instance~$(S,P)$. In the following we show a structural result of the
resulting~$\JordanCurve$ instances.

\begin{claim}
  \label{claim:no-syn-sol}
  If~$(S, P)$ is a $\ueopl$ instance with no syntactic solution, then
  the $\JordanCurve$ instance~$(B,R)$ resulting from the above
  reduction has no syntactic solution. Furthermore, any
  solution~$(t_r, t_b)$ of the instance~$(B,R)$ satisfies that
  \begin{enumerate}
  \item $t_r \notin \cT^R_d$ and $t_b \notin \cT^B_d$, and
  \item at least one of~$t_r$ or~$t_b$ is in an interesting time
    interval.
  \end{enumerate}
\end{claim}
\begin{proof}
  Let us begin by arguing that the $\JordanCurve$ instance~$(B,R)$
  contains no syntactic solution. By construction the two curves start
  and end at the correct positions. Furthermore, since every time
  interval is of size~$4D$, there are no discontinuities. This shows
  that~$(B,R)$ contains no syntactic solution.

  Consider a solution~$(t_r, t_b)$ of the instance~$(B, R)$. It cannot
  happen that~$t_r$ or~$t_b$ is contained
  in~$\cT_{\text{start}} \cup \cT_{\text{end}}$: all these timestamps
  are independently of the~$\ueopl$ instance~$(S,P)$ mapped to
  positions no other time interval may map to.

  Suppose the position~$R(t_r) = B(t_b)$ lies on row~$y_B$. Note that
  the red curve~$R$ crosses the row~$y_B$ only at positions between
  checkpoints~$\checkB_{i,j}$ and~$\checkB_{\nextT(i,j)}$ with the
  property that the interval~$\cT_{i,j}$ is \emph{not} dull. Hence any
  solution on row~$y_B$ involves two non-dull timestamps. By symmetry
  the same holds for any solution on row~$y_R$. If the
  intersection~$R(t_r) = B(t_b)$ lies on neither row~$y_B$ nor
  row~$y_R$, then clearly both timestamps are interesting and hence
  not dull. The first item follows.

  Regarding the second item, note that if both timestamps~$t_R, t_B$
  are not interesting, then~$R(t_r)$ lies on row~$y_R$
  whereas~$B(t_b)$ lies on row~$y_B$. Since the rows~$y_R$ and~$y_B$
  are distinct this is a contradiction; the statement follows.
\end{proof}

\begin{lemma}
  \label{lem:solutions}
  Suppose that~$(t_r, t_b)$ is a solution of a $\JordanCurve$
  instance~$(B,R)$ obtained by the above reduction from a~$\ueopl$
  instance~$(S,P)$. If~$t_r \in \cT_{i,j}$ and~$t_b \in \cT_{i',j'}$,
  then the following holds.
  \begin{enumerate}
  \item If~$\cT_{i,j}$ is not interesting, then~$(i, j)$ is a source
    or a sink.
    \label{item:ij_source_sink}
  \item If~$\cT_{i',j'}$ is not interesting, then~$(i', j')$ is a
    source or a sink.
    \label{item:ij_prime_source_sink}
  \item If~$\cT_{i,j}$ and~$\cT_{i',j'}$ are both interesting, then
    there are parallel active nodes.
    \label{item:parallel}
  \end{enumerate}
\end{lemma}

\begin{proof}
  Since~$\cT_{i,j}$ is not interesting, this implies that either the
  node~$(i,j)$ of the $\ueopl$ instance is (1) inactive, or (2) has no
  predecessor, or (3) has no successor. Since we may assume without
  loss of generality that the $\ueopl$ instance~$(S,P)$ has consistent
  circuits, it holds that a node is inactive if and only if it has no
  successor. Hence either~$(i,j)$ has no predecessor and is thus a
  source, or a sink in case it has no successor. This yields
  \cref{item:ij_source_sink}. The argument for
  \cref{item:ij_prime_source_sink} is symmetric.

  It remains to argue \cref{item:parallel}. Note that $|i - i'| = 1$,
  that is, the two nodes~$(i,j)$ and~$(i',j')$ are on consecutive rows
  in the $\ueopl$ instance. This holds because a curve on an in
  teresting interval ``connects'' the corresponding predecessor and
  successor that lie on the previous and the following row. As the two
  curves are interesting on these intervals it implies that both have
  valid predecessors and successors. Hence some of these are parallel.
\end{proof}

From \cref{lem:solutions} and \cref{claim:no-syn-sol} the correctness
of the reduction from \ueopl to \JordanCurve follows. This completes
the proof of \cref{theorem:ueopl_leq_jc}.

\small

\DeclareUrlCommand{\Doi}{\urlstyle{sf}}
\renewcommand{\path}[1]{\small\Doi{#1}}
\renewcommand{\url}[1]{\href{#1}{\small\Doi{#1}}}
\bibliographystyle{alphaurl}
\bibliography{bibliography}

\end{document}